\documentclass[12pt]{article}
\input{epsf.sty}
\usepackage[dvips]{graphicx}
\usepackage{mathrsfs}
\input{epsf.sty}
\usepackage{epsfig}
\usepackage{amsmath}
\usepackage{latexsym}
\usepackage{amssymb}
\usepackage{bm}
\usepackage{cite}
\usepackage{amsthm}
\usepackage[dvipsnames,usenames]{color}
\usepackage{graphicx}
\usepackage{times}
\usepackage{graphics,color}
\usepackage{times,url}
\usepackage{subfigure}
\usepackage{amsthm}
\usepackage{tikz}
\theoremstyle{definition}
\newtheorem{theorem}{Theorem}

\newtheorem{proposition}{Proposition}

\newtheorem{remark}{Remark}

 \graphicspath{/figures}

\makeatletter
\let\normalequation=\equation
\def\equation{\@ifnextchar[{\subequation}{\normalequation}}
\def\subequation[#1]#2{\@ifundefined{r@#1}%
  {\def\theequation{\bf ??#2}\@warning
    {Reference `#1' on page \thepage \space
     undefined}}{\edef\@tempa{\@nameuse{r@#1}}%
    \edef\theequation{\expandafter\@car\@tempa \@nil#2}}%
  \let\@currentlabel\theequation $$}
\makeatother

\graphicspath{{figs/}}

\begin{document}
\begin{center}
{\large\bf  A note on finite-time and fixed-time stability}
\footnote{This work is
jointly supported jointly
supported by  the National Natural
Sciences Foundation of China (Nos. 61273211 and 61273309), and the Program
for New Century Excellent Talents in University (NCET-13-0139).
}
\\[0.2in]
\begin{center}
Lu Wenlian \footnote{
Wenlian Lu is with the Centre for Computational
Systems Biology and School of Mathematical Sciences, Fudan University, People's Republic of China (wenlian@fudan.edu.cn)},
Liu Xiwei\footnote{Xiwei Liu is with Department of Computer Science and
Technology, Tongji University, and with the Key Laboratory of Embedded System and Service Computing,
Ministry of Education, Shanghai 200092, China. E-mail:
xwliu@tongji.edu.cn},
Tianping Chen\footnote{Tianping Chen is with the School of Computer
Sciences/Mathematics, Fudan University, 200433, Shanghai, China. \\
\indent ~~Corresponding author: Tianping Chen. Email:
tchen@fudan.edu.cn}
\end{center}
\end{center}

\begin{abstract}
In this letter,  by regarding finite-time stability as an inverse problem, we reveal the essence of finite-time stability and fixed-time stability. Some necessary and sufficient conditions are given. As application, we give a new approach for finite-time and fixed-time synchronization and consensus. Many existing results can be derived by the general approach.
\end{abstract}

Key words: Consensus; Time-varying topology; Event-triggered algorithm


\pagestyle{plain}

\section{Introduction}\quad
In many practical situations, stability over a finite time interval is of interests rather than the classic Lyapunov asymptotic stability, since it is more phyiscally realizable than concerning infinite time. There are two categories of concepts of stability over finite time interval. One is finite-time stability that means that the system converges within a finite time interval for any initial values; the other is fixed-time stability that means that the time intervals of convergence have a uniform upper-bounds for all initial values within the definitive domain. The previous works on this topic include Dorato 1961, Roxin 1966, Haimo 1986, Bhat and Bernstein 1998, 2000, Hong, et. 2002,. In particular, Lu and Chen 2005 presented analysis of the finite-time convergence for Cohen-Grossberg neural networks with discontinuous activation function.

The finite-time and fixed-time stability/convergence have been successfully applied in many fields. More related to the present letter, synchronization and consensus in networked systems have been attracting increasing interests (Lu, W.L and Chen 2004, 2006). Many recent literature were concerned with proposing schemes to realize finite-time synchronization/consensus .See Cort'{e}s 2006; Shen and Xia 2008, 2009; Xiao et al 2009; Jiang and Wang 2009; Wang and Hong 2010; Jiang and Wang 2011; Su et al 2012; Polyakov 2012; Parsegov et al 2012, 2013; Zhao et al 2015; Liu et al 2014, 2015; Wang et al 2014; Zou and Tie 2014a,b; Polyakov et al 2015; Zou 2015; Meng et al 2015 for reference. The main techniques in these works depends on the candidate Lyapunov functions as well as its cnvergence. In this letter, we propose a simple, novel and general technique to re-visit the problem of finite-time and fixed-time stability by regarding as an implicit inverse function of time and apply to the synchronization and consensus in networked system.

To exploit the idea, suppose a nonnegative scalar function $V(t)$ satisfies
\begin{align}\label{basic}
\dot{V}(t)=-\mu(V(t)),
\end{align}
where functions $\mu(V(t))>0$, $V(t)>0$; $\mu(0)=0$.

Because $\dot{V}(t)> 0$, $V(t)$ is decreasing. Therefore, the trajectory can also be written as
\begin{align}\label{basic1}
\dot{t}(V)=-\mu^{-1}(V)
\end{align}
(See the  following figure)
\begin{figure}
\centering
\includegraphics[width=0.3\textwidth]{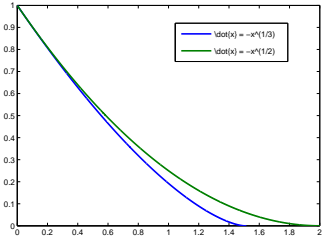}
\caption{Convergence behaviors of Eq.
\eqref{basic}}. \label{figSimu1}
\end{figure}
where $t(V)$ be the inverse function of $V(t)$. Then,
\begin{align}
t=-\int_{V(0)}^{V(t)}\mu^{-1}(V)dV
\end{align}
Therefore, the least time $t^{*}$ to make $V(t^{*}(V(0)))=0$ is
\begin{align*}
t^{*}(V(0))=\int_{0}^{V(0)}\frac{1}{\mu(V)}dV
\end{align*}

In summary, we have
\begin{proposition}
For system (\ref{basic}),
\begin{enumerate}
    \item ``$0$'' is a finite-time stable equilibrium for the system (\ref{basic}), i.e. there exists a time $t^{*}=t^{*}(V(0))$ depending on the initial value $V(0)$, such that $V(t)=0$, if $t\ge t^{*}(V(0))$, it is necessary and sufficient that the integral
\begin{align*}
t^{*}(V(0))=\int_{0}^{V(0)}\frac{1}{\mu(V)}dV
\end{align*}
is finite.
    \item ``$0$'' is a fixed-time stable equilibrium for the system (\ref{basic}), i.e. there exists a time $t^{*}$ independent of the initial value $V(0)$, such that $V(t)=0$, if $t\ge t^{*}$, it is necessary and sufficient that the integral
\begin{align*}
t^{*}=\int_{0}^{\infty}\frac{1}{\mu(V)}dV
\end{align*}
is finite.
\end{enumerate}
\end{proposition}
\begin{remark}
It is clear that finite-time convergence is an inverse problem: To find the time $\bar{t}$ so that $V(\bar{t})=0$. Therefore, instead of $V(t)$, we discuss the inverse function $t(V)$.
Previous results reveal that the finite-time convergence depends on the behavior of $\mu(V)$ in the neighborhood of $V=0$.
\end{remark}
\begin{remark}
Instead, the fixed-time convergence depends on the behavior of $\mu(V)$ at $V=0$ as well as the behavior of $\mu(V)$ at $\infty$.
\end{remark}


\begin{remark}
Geometrically, a system is stable is equivalent to that its trajectory $x(t)$ is with finite length in state space $x$. Instead, finite-time convergence means that the trajectory $x(t)$ is with finite length in time-state space $(x,t)$.
\end{remark}

In case $\mu(s)=s$, then
\begin{align*}
t=-\alpha^{-1}\int_{V(0)}^{V(t)}\frac{1}{V}dV
=-\alpha^{-1}\mathrm{log}\frac{V(t)}{V(0)},
\end{align*}
and
$$V(t)=V(0)e^{-\alpha t}.$$

It is clear that the integral
\begin{align*}
\int_{0}^{V(0)}\frac{1}{\mu(V)}dV=\infty
\end{align*}
Therefore, there is no $t^*$ such that $V(t^*)$=0. In fact, the system (\ref{basic}) is exponentially stable.

If
$\mu(s)=s^{p}$ with $p\ne 1$, then
\begin{align*}
t=\alpha^{-1}\int_{V(t)}^{V(0)}\frac{1}{V^{p}}dV
=\frac{V^{1-p}(0)-V^{1-p}(t)}{\alpha(1-p)}
\end{align*}
and
$$V(t)=[V^{1-p}(0)-\alpha(1-p)t
]^{\frac{1}{1-p}}$$

In case $p<1$, then $V(t)=0$, if
$$t\ge t^*=\alpha^{-1}\int_{0}^{V(0)}\frac{1}{V^{p}}dV=\frac{V^{1-p}(0)}{\alpha(1-p)}$$
which means that $V(t)$ converges to zero in finite-time.

On the other hand, in case $p>1$,
$$V(t)=\frac{1}{[\alpha(p-1)t
+V^{1-p}(0)]^{\frac{1}{p-1}}}$$
which means that $V(t)$ does not converges in finite time. Instead, it converges to zero with power rate $t^{-(p-1)}$ (see Chen. T. et.al, 2007).

\begin{figure}
\centering
\includegraphics[width=0.3\textwidth]{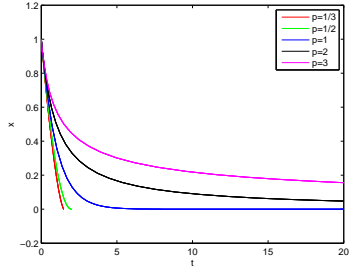}
\caption{Convergence behaviors for different index "p"}. \label{figSimu2}
\end{figure}

With similar approach, we have the following
\begin{theorem}
Suppose the Dini derivative of a nonnegative function $V(t)=V(z(t))$ satisfies
\begin{eqnarray}\label{time-triggered}
\dot{V}^{+}(t)\le\left\{\begin{array}{ccl}
-\mu_{1}(V(t))           & ;&\mathrm{if}~ 0<V<a\\
-\mu_{2}(V(t))      & ;&\mathrm{if}~ V\geq a
\end{array}\right.
\end{eqnarray}
for some constant $a>0$, where functions $\mu_{1}(V(t))>0$, $\mu_{2}(V(t))>0$, when $V(t)>0$; $\mu_{1}(0)=0$; and
\begin{align*}
\int_{0}^{a}\frac{1}{\mu_{1}(V)}dV=\omega_1<\infty,\\
\int_{a}^{\infty}\frac{1}{\mu_{2}(V)}dV=\omega_2<\infty
\end{align*}
for some constant $a>0$. Then $V(t)\equiv 0$ for all
$
t\ge \omega_1+\omega_2,
$
i.e., the fixed-time stability of ``$0$'' is realized.
\end{theorem}

\begin{proof}
In this case, no matter $V(0)\le 1$ or $V(0)\ge 1$, we can prove
\begin{align*}
t(0)-t(V(0))\le \omega_1+\omega_2.
\end{align*}
The proof is completed.
\end{proof}
\begin{remark}

System governed by (\ref{time-triggered}) is written as event-triggered system, which can also be written as following time-triggered model
\begin{eqnarray}
\dot{V}(t)\le\left\{\begin{array}{ccl}
-\mu_1(V(t)) & ;&\mathrm{if}~ t> \omega_2\\
-\mu_2(V(t))  & ;&\mathrm{if}~ t\leq \omega_2
\end{array}\right.
\end{eqnarray}
\end{remark}

\begin{remark}
The model discussed in Lu and Chen (2005) for Cohen-Grossberg neural networks with discontinuous activation functions can be  regarded as
$$\dot{V}(X(t))=\mu (V(X(t)),$$
where $\mu (V)$ is a discontinuous monotone-nondecresing function  with equilibrium $V^{*}$ lying in the discontinuity of the activation functions. In this case, finite-time convergence can be ensured (see Theorem 8 in Lu and Chen (2005)).
\end{remark}

\section{Applications: Finite-time and fixed-time synchronization and consensus}\label{syn}
In this section, we will apply the theoretical results given in previous section to finite-time and fixed-time synchronization and consensus, where the nodes are nonlinearly coupled and the network is a strongly connected undirected graph.


\subsection{Finite-time synchronization}
In (Lu and Chen 2004, 2006), and some other papers, the following linear coupled system
\begin{align}\label{modell}
\dot{x}_i(t)=&f(x_i(t))+\alpha \sum_{j=1}^Na_{ij}(x_j(t)-x_i(t)),
~~i=1,\cdots,N,
\end{align}
is discussed. By defining following useful reference node given in (Lu and Chen 2004, 2006):
\begin{align}
x^{\star}(t)=\frac{1}{N}\sum_{i=1}^Nx_i(t),
\end{align}
and following Lyapunov function (see Lu and Chen 2004, 2006)
\begin{align}\label{lya}
V(t)&=\frac{1}{2}\sum_{i=1}^N
(x_i(t)-x^{\star}(t))^T(x_i(t)-x^{\star}(t))\nonumber\\
&=\frac{1}{2N}\sum_{i,j=1}^N (x_i(t)-x_j(t))^T(x_i(t)-x_j(t))\nonumber\\
&=\frac{1}{2N}
\sum_{k=1}^n\sum_{i,j=1}^N|x_i^k(t)-x_j^k(t)|^{2},
\end{align}
it was proved that under some mild conditions, $\dot{V}(t)\le -\alpha V(t)$ for some constant $\alpha$. Therefore by previous result for $\mu(V)=V$, the convergence is exponential and not with finite-time.

To make the convergence finite-time, in this section, replacing (\ref{model1}), consider nonlinear coupled network with $N$ nodes:
\begin{align}\label{model}
\dot{x}_i(t)=&f(x_i(t))+\alpha \sum_{j=1}^Na_{ij}\Phi(x_j,x_i),
\end{align}
where scalars $\alpha>0,$ $x_i=(x_i^1,\cdots,x_i^n)^T\in R^n$, $i=1,\cdots,N$. Coupling matrix $A=(a_{ij})$ is symmetric and irreducible, with $a_{ij}\ge 0$, $i\ne j$.

Continuous function $f: R^n\rightarrow R^n$ satisfies: for any $U=[u^1,\cdots,u^n]^{T}\in R^n$, $V=[v^1,\cdots,v^n]^{T}\in R^n$, there exists a scalar $\delta>0$, such that
\begin{align}\label{quad}
(U-V)^T(f(U)-f(V))\le \delta (U-V)^T(U-V).
\end{align}

The nonlinear function $\Phi(\cdot,\cdot): R^n\times R^n\rightarrow R^n$ is defined as (see \cite{L09}):
\begin{align}
\Phi(U,V)=(\phi(u^1,v^1),\cdots,\phi(u^n,v^n))^T,
\end{align}
where $\phi(x,y)=sign(y-x)\frac{\mu(|x-y|^{2})}{|x-y|}$, and
the function $\mu(V)>0$ satisfies

1. semi-norm property $$\mu(V+U)\le c_{\mu}(\mu(U)+\mu(V))$$

2. $\frac{V}{\mu(V)}$ is monotone increasing when $V$ increasing.


Noticing $A$ is strongly connected and the semi-norm property for function $\mu$, we can find some constants $a_{\mu,A}$, $\bar{a}_{\mu,A}$ depending on the function $\mu$ and the coupling matrix $A$ such that
\begin{align}
&\sum_{k=1}^n\sum_{i,j=1}^Na_{ij}|x_i^k(t)-x_j^k(t)|
\phi(|x_i^k(t)-x_j^k(t)|)
\nonumber\\
&= \sum_{k=1}^n\sum_{i,j=1}^N a_{ij}\mu(|x_i^k(t)-x_j^k(t)|^{2})
\nonumber\\
&\ge a_{\mu,A} \sum_{k=1}^n\sum_{i,j=1}^N\mu(|x_i^k(t)-x_j^k(t)|^{2})
\nonumber\\
&\ge \bar{a}_{\mu,A} \mu(\sum_{k=1}^n\sum_{i,j=1}^N|x_i^k(t)-x_j^k(t)|^{2})
\end{align}

Based on previous preparations, we can give
\begin{theorem}\label{bx}
Suppose that irreducible matrix $A=(a_{ij})\in R^{N\times N}$ satisfies $a_{ij}=a_{ji}\ge 0$, $a_{ii}=0$. The coupling strength $\alpha$ is chosen such that
\begin{align}
\frac{\alpha \bar{a}_{\mu,A} }{2}>2\delta \frac{V(0)}{\mu(V(0))}.
\end{align}
Then, the  system (\ref{model}) reaches synchronization for all
\begin{align}
t\ge \bigg(\frac{\alpha \bar{a}_{\mu,A} }{2}-2\delta \frac{V(0)}{\mu(V(0))}\bigg)^{-1}\int_{0}^{V(0)}\frac{dV}{\mu(V)}
\end{align}
\end{theorem}
\begin{proof}
Differentiating the Lyapunov function (\ref{lya}), we have
\begin{align}
\dot{V}(t)=&\sum_{i=1}^N(x_i(t)-x^{\star}(t))^T(x_i(t)
-x^{\star}(t))^{\prime}\nonumber\\
=&\sum_{i=1}^N(x_i(t)-x^{\star}(t))^T(f(x_i(t))-f(x^{\star}(t)))
\nonumber\\
&+\alpha\sum_{i=1}^N(x_i(t)-x^{\star}(t))^T
\sum_{j=1}^Na_{ij}\Phi(x_j(t),x_i(t))\nonumber\\
=&V_1(t)+V_2(t),
\end{align}
where
\begin{align}\label{v1}
V_1(t)\le \delta\sum_{i=1}^N(x_i(t)-x^{\star}(t))^T(x_i(t)-x^{\star}(t))
=2\delta V(t),
\end{align}

By some algebra, we have
\begin{align}\label{v2}
V_2(t)=&\alpha\sum_{i=1}^N(x_i(t)-x^{\star}(t))^T
\sum_{j=1}^Na_{ij}\Phi(x_j(t),x_i(t))\nonumber\\
=&\alpha\sum_{i,j=1}^Na_{ij}x_i(t)^T
\Phi(x_j(t),x_i(t))
\nonumber\\
= &\frac{\alpha}{2}\sum_{i,j=1}^Na_{ij}(x_i(t)-x_j(t))^T
\Phi(x_j(t),x_i(t))\nonumber\\
 =&-\frac{\alpha }{2}\sum_{k=1}^n\sum_{i,j=1}^Na_{ij}|x_i^k(t)-x_j^k(t)|
\phi(x_j^k(t),x_i^k(t))
\nonumber\\\le & -\frac{\alpha \bar{a}_{\mu,A}}{2} \mu(\sum_{k=1}^n\sum_{i,j=1}^N|x_i^k(t)-x_j^k(t)|^{2})
\end{align}

Therefore,
\begin{align*}
\dot{V}(t)\le& 2\delta V(t)-\frac{\alpha \bar{a}_{\mu,A} }{2}\mu(V(t))\nonumber\\
=& -\bigg(\frac{\alpha \bar{a}_{\mu,A} }{2}-2\delta \frac{V(t)}{\mu(V(t))}\bigg)\mu(V(t))
\nonumber\\
\le& -\bigg(\frac{\alpha \bar{a}_{\mu,A} }{2}-2\delta \frac{V(0)}{\mu(V(0))}\bigg)\mu(V(t))\end{align*}
By Theorem 1, the proof is completed.
\end{proof}

\subsection{Fixed-time synchronization}
In this part, we consider fixed-time synchronization.

Suppose $A=(a_{ij})$ and $B=(b_{ij})$ with $a_{ij}\ge 0$ and $b_{ij}\ge 0, i\ne j$ are two symmetric and irreducible matrices.

Consider following system
\begin{align}\label{model1}
\dot{x}_i(t)=\left\{
\begin{array}{ll}
f(x_i(t))+\alpha \sum_{j=1}^Na_{ij}\Phi(x_j,x_i)&t\ge
t^{*}\\
f(x_i(t))+\beta\sum_{j=1}^Nb_{ij}\Psi(x_j,x_i); &t\le t^{*}
\end{array}\right.
\end{align}
where $\psi(x,y)=sign(y-x)\frac{\nu(|x-y|^{2})}{|x-y|}$, the function $\nu(V)>0$ satisfies
$\nu(V+U)\le c_{\nu}(\nu(U)+\nu(V))$
and $\frac{V}{\nu(V)}$ is monotone decreasing when $V$ increasing.

$\bar{b}_{\nu,B}$ is a constant depending on the function $\nu$ and the coupling matrix $B$ such that
\begin{align}
&\sum_{k=1}^n\sum_{i,j=1}^Nb_{ij}|x_i^k(t)-x_j^k(t)|
\psi(|x_i^k(t)-x_j^k(t)|)
\nonumber\\
&\ge \bar{b}_{\nu,B} \nu(\sum_{k=1}^n\sum_{i,j=1}^N|x_i^k(t)-x_j^k(t)|^{2})
\end{align}

Then, we have the following fixed-time synchronization result.
\begin{theorem}\label{main}
System (\ref{model1})
can reaches synchronization for all
\begin{align}\label{settling2}
t\ge \bigg(\frac{\alpha \bar{a}_{\mu,A} }{2}-2\delta
\frac{1}{\mu(1)}\bigg)^{-1}\int_{0}^{1}\frac{dV}{\mu(V)} +\bigg(\frac{\beta
\bar{b}_{\nu,B}}{2}\bar{b}_{\nu,B}-2\delta \frac{1}{\nu(1)}
\bigg)^{-1}\int_{1}^{\infty} \frac{dV}{\nu(V)}
\end{align}
\end{theorem}

\begin{proof}
We will first assume $V(0)>1$ and $x_{i}(t)$, $i=1,\cdots,N,$ satisfies
\begin{align}
\dot{x}_i(t)=
f(x_i(t))+\beta\sum_{j=1}^Nb_{ij}\Psi(x_j,x_i).
\end{align}

In this case, differentiating the Lyapunov function, we have
\begin{align}
\dot{V}(t)=V_{1}(t)+V_{3}(t),
\end{align}
where $V_{1}(t)$ is same as in Theorem 5, and
\begin{align}\label{v3}
V_3(t)=&\beta\sum_{i=1}^N(x_i(t)-x^{\star}(t))^T
\sum_{j=1}^Nb_{ij}\Psi(x_j(t),x_i(t))
\nonumber\\\le & -\frac{\beta \bar{b}_{\nu,B}}{2} \nu(\sum_{k=1}^n\sum_{i,j=1}^N|x_i^k(t)-x_j^k(t)|^{2})
\end{align}
Then
\begin{align}
\dot{V}(t) &\le -\frac{\beta \bar{b}_{\nu,B}}{2}\nu(V(t))+2\delta V(t) \le
-(\frac{\beta
\bar{b}_{\nu,B}}{2}-2\delta\frac{V(t)}{\nu(V(t))})\nu(V(t))\label{newin1}
\end{align}
and $V(t^{*})\le 1$, where
\begin{align}
t^{*}= \bigg(\frac{\beta \bar{b}_{\nu,B}}{2}-2\delta
\frac{1}{\nu(1)}\bigg)^{-1}\int_{1}^{\infty} \frac{dV}{\nu(V)}
\end{align}

Combining with Theorem \ref{bx}, one can get that the fixed-time
synchronization is finally realized, and the settling time is also given as
(\ref{settling2}).
\end{proof}

\subsection{Finite-time and fixed-time consensus}
It is clear that in case $f(\cdot)=0$, $n=1$, the finite-time and fixed-time synchronization problem becomes the finite-time and fixed-time consensus problem. As special examples of previous section, we consider following nonlinear consensus models
\begin{align}\label{consensus}
\dot{x}_i(t)= \sum_{j=1}^Na_{ij}\Phi(x_j,x_i),
\end{align}
and
\begin{align}\label{consensus1}
\dot{x}_i(t)=\left\{
\begin{array}{ll}
\sum_{j=1}^Na_{ij}\Phi(x_j,x_i)&t\ge
t^{*}\\
\sum_{j=1}^Nb_{ij}\Psi(x_j,x_i); &t\le t^{*}
\end{array}\right.
\end{align}

\begin{theorem}
The  system  (\ref{consensus}) reaches finite-time consensus, i.e., $x_{i}(t)=x_{j}(t)$ for all $i,j=1,\cdots,N$ and
\begin{align}
t\ge \bar{a}_{\mu,A}^{-1}\int_{0}^{V(0)}\frac{dV}{\mu(V)}
\end{align}\end{theorem}

\begin{theorem}\label{main2}
Denote
\begin{align}
t^{*}= \bigg(\frac{\beta \bar{b}_{\nu,B}}{2} \bigg)^{-1}\int_{1}^{\infty}
\frac{dV}{\nu(V)}
\end{align}
The  system  (\ref{consensus1}) reaches fixed-time consensus for all 
\begin{align}
t\ge t^{*}+\bar{a}_{\mu,A}^{-1}\int_{0}^{V(0)}\frac{dV}{\mu(V)}
\end{align}

\end{theorem}


\section{Conclusion}\label{con}
In this letter, by regarding finite-time stability as an inverse problem, we reveal the essence of finite-time stability and fixed-time stability. Some necessary and sufficient conditions are given. As application, we give a new approach for finite-time and fixed-time synchronization and consensus and some new results are given, too. As direct consequences, many existing results can be derived by the general approach.


\noindent{\bf\Large References}
\begin{description}
\item
Dorato, P. (1961).
Short time stability in linear time-varying systems
Proc. IRE Int. Convention Record Part 4, 83--87.

\item
Roxin, E. (1966).
On Finite Stability in Control Systems, {\it SIAM}, 49(9), 1520--1533.

\item
Haimo, V. T., (1986). Finite time controllers, {\it SIAM J. Control Optim.}, 24(4), 760-770.

\item
Bhat S.~P., and Bernstein, D.~S., (1998) Continuous finite-time stabilization of the translational and rotational double integrators, \emph{IEEE Trans. Autom. Control}, vol. 43, no. 5, pp. 678-682, May.
\item
Bhat S.~P., and Bernstein, D.~S., (2000) Finite-time stability of continuous autonomous systems, \emph{SIAM J. Control Optim.}, vol. 38, no. 3, pp. 751-766, Mar.
\item
Hong, Y.~G., Xu, Y.~S., and Huang, J., (2002) Finite-time control for robot manipulators, \emph{Syst. Control Lett.}, vol. 46, no. 4, pp. 243-253, Jul.

\item
Lu, W.~L., and  Chen, T.~P., (2004) ``Synchronization analysis of linearly coupled networks of discrete time systems,'' \emph{Physica D}, vol.198, pp. 148-168.


\item
Lu, W.~L., and  Chen, T.~P., (2005). Dynamical behaviors of
 Cohen-Grossberg neural networks with discontinuous activation
 functions," Neural Networks, 18, 231-242
\item
Lu, W.~L., and  Chen, T.~P., (2006) ``New approach to synchronization analysis of linearly coupled ordinary differential systems,'' \emph{Phys. D, Nonlinear Phenomena}, vol. 213, no. 2, pp. 214-230.

\item
Moulay, E., Perruquetti, W., (2006). Finite time stability and stabilization of a class of continuous systems, J. Math. Anal. Appl. 323, 1430¨C1443

\item
Cort\'{e}s, J., (2006). Finite-time convergent gradient flows with applications to network consensus, \emph{Automatica}, vol. 42, no. 11, pp. 1993-2000, Nov.

\item
Chen., T., and Wang., L., (2007).
"Power-Rate Global Stability of Dynamical Systems With Unbounded
Time-Varying Delays", IEEE Transactions on Circuits and Systems¡ªII:
Express Briefs, 54(8), 705-709

\item
Shen, Y.~J.,  and Xia, X.~H., (2008). Semi-global finite-time observers for nonlinear systems, \emph{Automatica}, vol. 44, no. 12, pp. 3152-3156, Dec.

\item
Shen, Y.~J., and  Huang, Y.~H., (2009) Uniformly observable and globally lipschitzian nonlinear systems admit global finite-time observers, \emph{IEEE Trans. Autom. Control}, vol. 54, no. 11, pp. 2621-2625, Nov.

\item
Xiao, F.,  Wang, L.,  Chen, J., and Y. P. Gao, (2009). Finite-time formation control
for multi-agent systems, Automatica, vol. 45, no. 11, pp. 2605-2611, Nov.

\item
 Jiang, F. C., and  Wang, L., (2009) Finite-time information consensus for multiagent
systems with fixed and switching topologies, Physica D, vol. 238,
no. 16, pp. 1550-1560, Aug.

\item
Wang X. L., and  Hong, Y. G., (2010). Distributed finite-time $\xi$-consensus algorithms for multi-agent systems with variable coupling topology, J. Syst. Sci. Complex, vol. 23, no. 2, pp. 209-218.

\item
Wang, L., and  Xiao, F., (2010). Finite-time consensus problem for network of dynamic agents, IEEE Trans. Autom. Control, vol. 55, no. 4, pp. 950-955.
\item
Jiang F. C., and  Wang, L., (2011).  Finite-time weighted average consensus with respect to a monotonic function and its application, Syst. Control Lett., vol. 60, no. 9, pp. 718-725.

\item
Sun, F. L.,  Chen, J. C.,  Guan, Z. H.,  Ding, L., and  Li, T., (2012) Leader following
finite-time consensus for multi-agent systems with jointly reachable
leader, Nonlinear Anal.-Real, vol. 13, no. 5, pp. 2271-2284.
\item
 Polyakov, A., (2012) Nonlinear feedback design for fixed-time stabilization of
linear control systems, IEEE Trans. Autom. Control, vol. 57, no. 8, pp.
2106-2110.

\item
 Parsegov, S.,  Polyakov, A. and  Shcherbakov, P., (2012) Nonlinear fixed-time
control protocol for uniform allocation of agents on a segment, in Proc.
51st IEEE Conf. on Decision and Control, Maui, Hawaii, USA, Dec.
2012, pp. 7732-7737.

\item
Parsegov, S.,  Polyakov, A. and  Shcherbakov, P., (2013) Fixed-time consensus
algorithm for multi-agent systems with integrator dynamics, in Proc. 4th
IFAC Workshop Distrbuted Estimation and Control in Networked System,
Koblenz, Germany, pp. 110-115.

\item
Zhao, Y.,  Duan, Z. S., and  Wen, G. H., (2015) Finite-time consensus for second order
multi-agent systems with saturated control protocols, IET Control
Theory Appl., vol. 9, no. 3, pp. 312-319.

\item
Liu, X. Y.,  Ho, D. W. C.,  Yu, W. W., and  Cao, J. D., (2014) A new switching
design to finite-time stabilization of nonlinear systems with applications
to neural networks, Neural Networks, vol. 57, pp. 94-102.

\item
Liu, X. Y., Lam, J.,  Yu, W. W. and  Chen, G.,  Finite-time consensus of
multiagent systems with a switching protocol, IEEE Trans. Neural Netw.
Learn. Syst., DOI: 10.1109/TNNLS.2015.2425933.

\item
Wang, X. Y.,  Li, S. H. and  Shi, P., (2014) Distributed Finite-time containment
control for double-integrator multiagent systems, IEEE Trans. Cybern.,
vol. 44, no. 9, pp. 1518-1528.

\item
Zuo Z. Y., and  Tie, L., (2014) A new class of finite-time nonlinear consensus
protocols for multi-agent systems, Int. J. Control, vol. 87, no. 2, pp. 363-370.

\item
Zuo Z. Y., and  Tie, L., Distributed robust finite-time nonlinear
consensus protocols for multi-agent systems, Int. J. Syst. Sci., doi:
10.1080/00207721.2014.925608.

\item
Polyakov, A.,  Efimov, D. and  Perruquetti, W., (2015) Finite-time and fixed-time stabilization: Implicit Lyapunov function approach, Automatica, vol. 51, pp. 332-340.

\item
Zuo Z. Y., (2015) Non-singular fixed-time terminal sliding mode control of
non-linear systems, IET Control Theory and Applications, vol. 9, no. 4,
pp. 545-552.

\item

Meng, D. Y.,  Jia, Y. M. and  Du, J. P., Finite-time consensus for
multiagent systems with cooperative and antagonistic interactions, IEEE
Trans. Neural Netw. Learn. Syst., DOI: 10.1109/TNNLS.2015.2424225

\end{description}

\end{document}